\documentclass[%
 reprint,
 amsmath,amssymb,
 aps,
]{revtex4-2}

\usepackage{orcidlink}
\usepackage{graphicx}% Include figure files
\usepackage{dcolumn}% Align table columns on decimal point
\usepackage{bm}% bold math
\usepackage{amsmath,amssymb,amsthm,easybmat,verbatim}

\usepackage{color}

\newtheorem{corollary}{Corollary}

\newtheorem{proposition}{Proposition}
\newtheorem{theorem}{Theorem}

\usepackage{physics}
\usepackage[caption=false]{subfig}
\usepackage{hyperref}

\begin{document}

\preprint{APS/123-QED}

\title{Quantum Distribution Error Mitigation via the Circulant Structure of Pauli Noise}% Force line breaks with \\

\author{Alvin Gonzales$^1$ \orcidlink{0000-0003-1635-106X}}
\email{agonzales@anl.gov}
\affiliation{$^1$Mathematics and Computer Science Division, 
Argonne National Laboratory, Lemont, IL, USA}

% \collaboration{CLEO Collaboration}%\noaffiliation

\date{\today}% It is always \today, today,

\begin{abstract}
This work introduces distribution error mitigation (DEM), which mitigates the error in the output distribution of a quantum circuit. We provide a rigorous theoretical foundation. If the composite noise affecting the circuit is a Pauli channel, the ideal output distribution and noisy distribution in the standard basis are related by a stochastic matrix. This system is described by a XOR convolution (the matrix is recursive 2 by 2 block circulant) between a noise vector and the ideal distribution. The noisy output distribution can be corrected to the ideal output distribution via a Fast Walsh-Hadamard Transform. We introduce a tomography method to approximate the noise vector, which requires sampling of only one logical circuit. The quantum overhead of DEM requires sampling of only two logical circuits. We provide techniques to scale the application of DEM efficiently. Accuracy bounds are provided. The approach is tested with quantum hardware executions consisting of 20-qubit and 30-qubit GHZ state preparation, 5-qubit Grover, 6-qubit and 10-qubit quantum phase estimation, and 10-qubit and 20-qubit Dicke state preparation circuits. DEM dramatically improves the accuracies of the output distributions for all demonstrations. For 30-qubit GHZ state preparation, a corrected distribution fidelity of 97.7\% is achieved from an initial raw fidelity of 23.2\%.
\end{abstract}

\maketitle
\makeatother
% \newtheorem{definition}{Definition}[section]
% \newtheorem{assumption}{Assumption}[section]
% \newtheorem{theorem}{Theorem}[section]
% \newtheorem{lemma}{Lemma}[section]
% \newtheorem{conjecture}{Conjecture}[section]
% \newtheorem{property}{Property}[section]
% \newtheorem{corollary}{Corollary}[section]

%

% \usepackage{biblatex}
% \addbibresource{dqva.bib}

%

% \newcommand{\Field}[0]{\symbb{F}}
% \DeclareMathOperator{\join}{join}
% \DeclareMathOperator{\sdwidth}{width}
% \DeclareMathOperator{\leaves}{lvs}
% \DeclareMathOperator{\gcut}{cut}
% \DeclareMathOperator{\emb}{emb}
% \DeclareMathOperator*{\argmin}{arg\ min}
% \DeclareMathOperator*{\argmax}{arg\ max}
% % \newcommand{\ket}[1]{\lvert #1 \rangle}

% \newcommand{\sminus}[0]{\scalebox{0.75}[1.0]{\( - \)}}
% \newcommand{\union}[0]{\cup}
% \newcommand{\intersect}[0]{\cap}
% \DeclareMathOperator*{\bigast}{\scalerel*{\ast}{\sum}}

\section{Introduction}
Quantum error correction codes (QECC) \cite{nielsen2011quantumCompAndQuantInfo, shor_1996FaultTolQuantComp, gottesman1997stabilizercodesquantumerrorCorr} are the standard approach for achieving reliable quantum computing. In fault tolerant quantum error correction, we can suppress the error rates to arbitrarily low levels provided that the physical error rates are below the code threshold \cite{ Aharonov_2008FaultTolQuantCompWithConstErrRate}.  This process imposes a large overhead on quantum computing. While progress has been made \cite{Acharya2024_QECBelowTheSurfCodeThresh}, fault tolerant quantum devices with many logical qubits are still far from reality with some estimates projecting that thousands of physical qubits are required for each logical qubit \cite{Fowler_2012_SurfCodeTowardsPracLargeScaleQuantComp, Chamberland_2017OverheadAnalysisOfUniversalConcatQuantCodes, Campbell_2017_RoadsTowardsFaultTolUnivQuantComp}. In its place, a suite of quantum error mitigation techniques are commonly used \cite{Maciejewski_2020MitigOfReadoutNoiseInNearTermQuantDevices, Jattana_2020_generalErrorMitigForQuantCircs, Mcclean_2017HybridQuantClassHeirForMitigOfDecohAndDeterOfExciteStates, Li_2017EffVarQuantSimIncActErrorMinimZNE, Temme_2017ErrorMitigForShortDepthQuantCircs, Debroy_2020ExtendFlagGadgetsForLowOverCircVer,  Gonzales_2023QEMByPCS, Huggins_2021VirtDistillForQuantErrMitig}. Additionally, error mitigation will likely be important in achieving quantum advantage \cite{aharonov2025importanceerrormitigationquantum}.

% In this work, we present an alternative and much simpler approach when the goal is to correct the output distribution from the quantum computer. This is in contrast to 
% In QECC the goal is to correct the logical quantum state exactly. 
% In most quantum computing scenarios, the goal is to generate an accurate resulting output distribution or expectation value. Thus, many types of errors are inconsequential. For instance, the Bell states $\ket{\Phi^{\pm}}=\frac{1}{\sqrt{2}}(\ket{00}\pm\ket{11})$ generate the same distribution in the standard basis and the relative phase is irrelevant. 
Error mitigation techniques usually target correcting the estimate of the expectation value of an observable. Here, we investigate correcting the measured output distribution, although our approach can also be used to improve observable measurement outcomes.

This work presents  distribution error mitigation (DEM). Any unitary quantum circuit can be described by an ideal circuit followed by a Kraus error channel and measurement in the standard basis. Provided that the composite noise channel affecting the circuit is Pauli, the noisy distribution and the ideal distribution are related by a stochastic matrix $A$. This system is described by an exclusive OR (XOR) convolution between a noise vector $\vec{a}$ and the ideal probability vector $\vec{x}$. Thus, $A$ has a recursive 2$\times$2 block circulant structure \cite{Rezghi_2011DiagOfTensWithCircStruct} and is  described by any one of its columns. The analytical results presented here are different from known properties of Pauli channels which relate the coefficients of two Pauli channel representations with a Fourier Transform \cite{chen_2022QuantAdvForPauliChanEst}. Provided we know the noise vector $\vec{a}$, this relation allows us to correct the noisy distribution via a Fast Walsh-Hadamard Transform (FWHT). Thus, determining $\vec{a}$ is the key hurdle. We present a simple noise vector tomography that allows us to estimate $\vec{a}$ with a single logical circuit.  

Thus, DEM only requires execution of 2 logical circuits. Biasing to Pauli noise is performed with randomized compiling \cite{Wallman_2016NoiseTailForScalQCViaRandmizedCompiling}. We also provide accuracy bounds for DEM and algorithms that allow us to apply DEM to a large number of qubits. The DEM approach is tested with quantum hardware executions consisting of state preparation circuits for 20-qubit and 30-qubit GHZ states and 10-qubit and 20-qubit Dicke states with 1 excitation (Dicke 10-1 and Dicke 20-1 states are denoted as $\ket{D^{10}_1}$ and $\ket{D^{20}_1}$, respectively), 6-qubit and 10-qubit quantum phase estimation, and 5-qubit Grover search. DEM dramatically improves the fidelities of all demonstrations. For the 30-qubit GHZ state preparation circuit, the output distribution is corrected to 97.7\% fidelity from the initial raw fidelity of 23.2\%. The 5-qubit Grover search circuit consists of 582 CZ gates and DEM achieves a corrected fidelity of 74.9\% from an initial raw fidelity of 10.2\%. A summary of the fidelities and gate counts for the quantum hardware executions are provided in Tables~\ref{tab:fidelities} and \ref{tab:gate_counts}, respectively. Finally, we discuss open problems with DEM.

\section{Background}
Let $X, Y,$ and $Z$ denote the Pauli matrices. The Pauli group is defined as
\begin{align}
    \mathcal{P}=\{I,X,Y,Z\}^{\otimes n}\times\{\pm 1, \pm i\}.    
\end{align}
An important class of unitary operations are Clifford operations, which maps the Pauli group into itself \cite{nielsen2011quantumCompAndQuantInfo}.
% In measurement error mitigation, the noisy distribution from quantum  the ideal output distribution (absence of measurement errors, but contains other errors such as gate errors)  by a stochastic matrix \cite{Nation_2021ScalableMitigOfMeasErrsOnQuantComp}, the ideal distribution $p_I$  and the noisy distribution $p_N$ (distribution from the full circuit)

Randomized compiling \cite{Wallman_2016NoiseTailForScalQCViaRandmizedCompiling} is a method that  biases gate errors towards Pauli channels via twirling \cite{Bennett_1996MixedStEntAndQEC, Cai_2019ConstructSmallerPauliTwirlingSetsForArbErrChann}.
Twirling is defined as 
\begin{align}
    \mathcal{E}'(\rho)=\frac{1}{\abs{T}}\sum_{i,(V\in T)}VE_iV^\dagger \rho VE_i^\dagger V^\dagger,
\end{align}
where $T$ is the Twirling set and the $E_i$s are Kraus operators of the original noise channel.

Next, given a payload circuit, a noise estimation circuit (NEC) is a constructed circuit with a similar structure \cite{Urbanek_2021MitigDepolNoiseOnQCWithNoiseEstimationCircs}. The noise estimation circuit is typically used to estimate the errors affecting a quantum circuit.  Finally, the Hellinger fidelity between two distributions $p=\{p_i\}$ and $q=\{q_i\}$ over outputs $i$ is defined as 
$F(p,q)=\left(\sum_i\sqrt{p_iq_i}\right)^2$.

The total-variation (TV) distance is an important distance measure between probability distributions. It is defined as 
\begin{align}
    D_{TV}(p,q)=\frac{1}{2}\norm{p,q}_1,
\end{align}
where $p$ and $q$ are probability vectors \cite{Maciejewski_2020MitigOfReadoutNoiseInNearTermQuantDevices}. This is related to the distance between POVMs $\textbf{M}$ and $\textbf{N}$, where $\textbf{M}$ and $\textbf{N}$ are block vectors with elements of the corresponding POVM operators.
\begin{align}
    D_{OP}(\textbf{M}, \textbf{N})=\text{max}_\rho D_{TV}(\textbf{p}^\textbf{M},\textbf{p}^\textbf{N}),
\end{align}
where we maximize over all quantum states and $\textbf{p}^{\textbf{M}/\textbf{N}}$ are the probability distributions corresponding to the state $\rho$ and POVM $\textbf{M}/\textbf{N}$.
The next two sections describe the theoretical setup and detailed derivations of important properties of we will rely on. Detailed proofs for Prop.~\ref{prop:stochastic_relation} and Prop.~\ref{prop:AtildeSymmetry} are provided in the Appendix.

% \section{Results}

\section{Theory}
\subsection{Quantum Circuit Error Assignment Matrix}

For an arbitrary noisy unitary quantum circuit, we can write the evolution as an ideal unitary channel followed by a composite noise channel. Assuming that the qubits are initialized in a product state with the environment, the evolution of the qubits from the initial state $\rho$ to the final state $\rho'$ is given by a completely positive map $\mathcal{K}(\rho)=\sum_iK_i\rho K_i^\dagger$. Thus,
\begin{align}
    \notag\rho'=\mathcal{K}(\rho)=\sum_i(K_iU^\dagger) U\rho U^\dagger (UK_i^\dagger)\\
    \label{eq:tilde_rho_in}=\mathcal{E}\circ \mathcal{U}(\rho)=\sum_iE_iU\rho U^\dagger E_i^\dagger=\sum_{i}E_i\tilde \rho E_i^\dagger,
\end{align}
where $U$ represents the ideal circuit, $\mathcal{E}(\omega)=\sum_iE_i\omega E_i^\dagger$ is the noise channel, $E_i=K_iU^\dagger$, and $\tilde\rho=\mathcal{U}(\rho)=U\rho U^\dagger$. The rotations corresponding to the measurement basis are included in $U$ and the measurement basis can be assumed to be in the standard basis. This description is general for completely positive  channels $\mathcal{K}$. Noise from the measurements can be composed into $\mathcal{E}$. Note that $\tilde\rho$ is the ideal (noise free) state generated by the quantum circuit. A quantum evolution can be equivalently described by a superoperator matrix $A$ \cite{Sudarshan_1961StochasticDynOfQuantMechSystems}, which for $n$ qubits is $4^n\times 4^n$.

A simpler scenario arises when $\mathcal{E}$ is a Pauli noise channel. In this scenario, Eq. \eqref{eq:tilde_rho_in} becomes
\begin{align}\label{eq:pauli-bias}
    \rho'=\sum_{i}\chi_{i}P_i\tilde \rho P_i,
\end{align}
where $\chi_i$ is real and $P_i$ is from the set of $+1$ elements of the Pauli group. In the following analytical calculations, we constrain $\mathcal{E}$ to Pauli.

\begin{proposition}\label{prop:stochastic_relation}
    Let the composite noise channel affecting the quantum circuit be Pauli. Then the noisy output distribution $\vec{z}$ and the ideal output distribution $\vec{x}$ in the standard basis are related by a stochastic matrix $A$, i.e., $A\vec{x}=\vec{z}$.
\end{proposition}

A detailed proof is provided in the Appendix. A similar result was proven for the more restricted case of measurement errors in Ref.~\cite{Geller_2020_rigorousMeasErrCorrec}. $A$ has $2^n\times 2^n$ elements and consists of only the relevant terms that affect the distribution in the standard basis for a Pauli a noise channel.
% Notice that we do not need to characterize $A$, but only the submatrix $A$ of $A$. 
Next, we describe symmetry properties of $A$. 
% that we use to characterize it with a single circuit and perform inversion in an efficient manner.
\begin{proposition}\label{prop:AtildeSymmetry}
    The elements of $A$ are related by
    \begin{align}\label{eq:tilde_a_symmetry}
        \bra{j}A P_x\ket{k}
    = \bra{j}P_xA\ket{k} \quad \forall \ket{j},\ket{k}, P_x
    \end{align}
    where $P_x$ is a Pauli X string (e.g., $X\otimes I$) and $\ket{j}$ and $\ket{k}$ are standard basis elements.
\end{proposition}
A detailed proof is provided in the Appendix. The proof follows from the fact that Pauli channels commute.

Proposition \ref{prop:AtildeSymmetry}, implies that $A$ is completely characterized by \textit{any one} of its columns, since all of the columns of $A$ have the same elements, but reordered. Therefore, we only have to characterize one column of $A$ to characterize all of $A$. Some important properties of $A$ immediately follow from Prop.~\ref{prop:AtildeSymmetry}.
\begin{corollary}
    $A$ is symmetric.
\end{corollary}
\begin{proof}
    This immediately follows from Prop.~\ref{prop:AtildeSymmetry} by setting $\ket{j}=\ket{k}$. Then $\bra{k}A\ket{k'}=\bra{k'}A\ket{k}$ $\forall \ket{k'},\ket{k}$, where $\ket{k'}=P_x\ket{k}.$
\end{proof}

\begin{corollary}
    $A$ is doubly stochastic (i.e., each row and each column sums to 1).
\end{corollary}
\begin{proof}
   The first column $\vec{a}$ of $A$ sums to 1. From Prop.~\ref{prop:AtildeSymmetry}, each element of $\vec{a}$ appears once in each column. Additionally, Prop.~\ref{prop:AtildeSymmetry} implies that the row position of each element is different for each column.
\end{proof}

\begin{theorem}
    The system $\vec{z}=A \vec{x}$ is described by the XOR convolution
    \begin{align}
        z_i=\sum_j a_{i\oplus j}x_j,
    \end{align}
    where $\vec{a}$ is the first column vector of $A$ and $\oplus$ is the XOR operation.
\end{theorem}
\begin{proof}
    In index notation, we have $z_i=\sum_jA_{ij}x_j$. Using $i,j$ in the standard basis, we have
    \begin{align}\label{eq:systemInIndexNota}
        A_{ij}=\bra{i}A\ket{j}=\bra{i}A P_x\ket{0},
    \end{align}
    where $\ket{j}=P_x\ket{0}$. Notice that $\bra{i}P_x=\bra{i\oplus j}$. Then from Prop. \ref{prop:AtildeSymmetry},
    \begin{align}\label{eq:tildeAInIndexNota}
        A_{ij}=\bra{i}P_xA\ket{0}=\bra{i\oplus j}\vec{a}=a_{i\oplus j}.
    \end{align}
    Substituting Eq.~\eqref{eq:tildeAInIndexNota} into Eq.~\eqref{eq:systemInIndexNota}, we get
    \begin{align}
        z_i=\sum_ja_{i\oplus j}x_j.
    \end{align}
\end{proof}

We refer to $\vec{a}$ as the noise vector. From convolution theory, this implies that 
\begin{corollary}
$A$ is recursive 2$\times$2 block circulant and is diagonalized by the Walsh-Hadamard Transform.
\end{corollary}
Note that a $2\times 2$ circulant matrix has the form $M=\begin{bmatrix}
    b & c\\
    c & b
\end{bmatrix}.$

\subsection{Inversion of $A$}
In general $A$ is not guaranteed to be invertible.  
% When it is invertible, we can correct exactly and when it not invertible, we can correct approximately. We discuss this in more detail in Sec.~\ref{sec:statstical_analysis}.
Due to the recursive 2$\times$2 block circulant structure of $A$,  we can easily construct its inverse or pseudo-inverse \cite{Rezghi_2011DiagOfTensWithCircStruct}. $A$ is diagonalized by the tensor product of $2\times 2$ Fourier Transform matrices, i.e, $H^{\otimes n}$ which is also known as the Hadamard transform. The system of equations $A\vec{x}=\vec{z}$, where $\vec{x}$ is the ideal counts and $\vec{z}$ is the noisy distribution of the raw payload circuit, is solved by
\begin{align}\label{eq:inversion_fwht}
    \vec{x}=\text{IFWHT}\left(\text{FWHT}(\vec{z})./\text{FWHT}(\vec{a})\right),
\end{align}
where $\vec a$ is the first column of $A$, $./$ is element wise division, FWHT is the Fast Walsh-Hadamard transform, and IFWHT is the inverse FWHT \cite{Rezghi_2011DiagOfTensWithCircStruct}. For elements where a division by zero occurs, we can assign an element output of zero, which is equivalent to performing the pseudo-inverse. Eq.~\eqref{eq:inversion_fwht} is much faster than performing matrix inversion. The number of operations required is $O(m\log(m))$, where $m$ is the dimension of the array. In contrast, matrix inversion scales $O(m^3)$.

% We first describe the high level idea of DEM. 
\section{Distribution Error Mitigation}
At a high level distribution error mitigation (DEM) works by following the previous discussion by 1.) approximating the noise vector $\vec{a}$. 2.) correcting the noisy distribution $\vec{z}$ by leveraging the circulant structure. DEM makes several assumptions that will be explicitly stated along the way and we will introduce methods that bring us closer to these assumptions. Additionally, improvements to DEM will be introduced.

Notice that $\vec{z}$ is determined by sampling the payload circuit. Assuming we know $\vec a$, we can correct the distribution with Eq.~\eqref{eq:inversion_fwht} (scalability will be discussed later). Thus, the main difficulty in applying DEM is determining $\vec a$.  
% In the demonstrations, using Eq.~\eqref{eq:inversion_fwht} takes only a few seconds with array sizes of up to $2^{15}$ elements.

\subsection{Noise Vector Tomography with 1 Logical Circuit}
We first introduce a simple characterization scheme for $\vec{a}$ using a noise estimation circuit (NEC). We refer to this method as vanilla NEC. Note that the characterization of $\vec{a}$ is not process tomography of the full Pauli channel, since $Z$ errors have no effect on the standard distribution. \textbf{Assumption 1: the composite noise channel (see Eq.~\eqref{eq:tilde_rho_in}) affecting the payload circuit is Pauli}. We can get closer to this assumption via randomized compiling. To perform characterization, we first transpile the circuit so that it complies with the gate set and connectivity of the quantum hardware device. Call this transpiled circuit the payload circuit. Without a loss of generality, let the gate set be given by CZ, SX, RZ, and X. Notice that the only gate capable of creating a superposition is SX
\begin{align}
    \text{SX}=\frac{1}{2}\begin{bmatrix}
    1 &-i\\
    -i &1
    \end{bmatrix}.
\end{align}

We construct a NEC from the payload circuit by replacing SX gates with X gates while keeping everything else the same. An example of the NEC construction process is provided in Figs.~\ref{fig:payload_circ_ex} and \ref{fig:nec_ex}.
\begin{figure}
    \centering
    \subfloat[\label{fig:payload_circ_ex} Payload circuit.]{\includegraphics[width=0.45\textwidth]{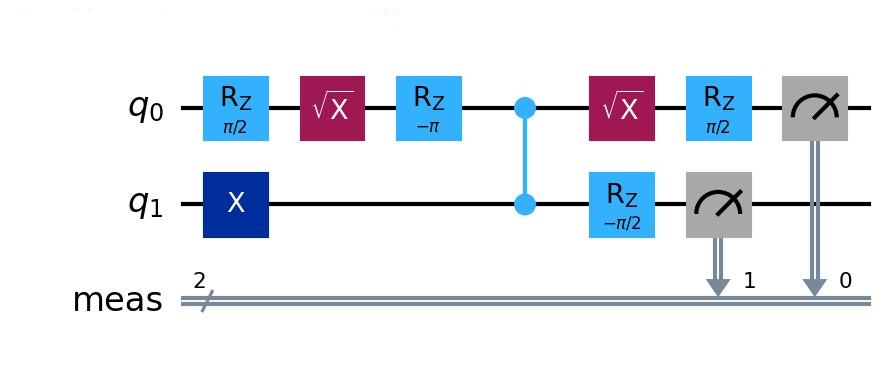}}

    \subfloat[\label{fig:nec_ex} Noise estimation circuit.]{\includegraphics[width=0.45\textwidth]{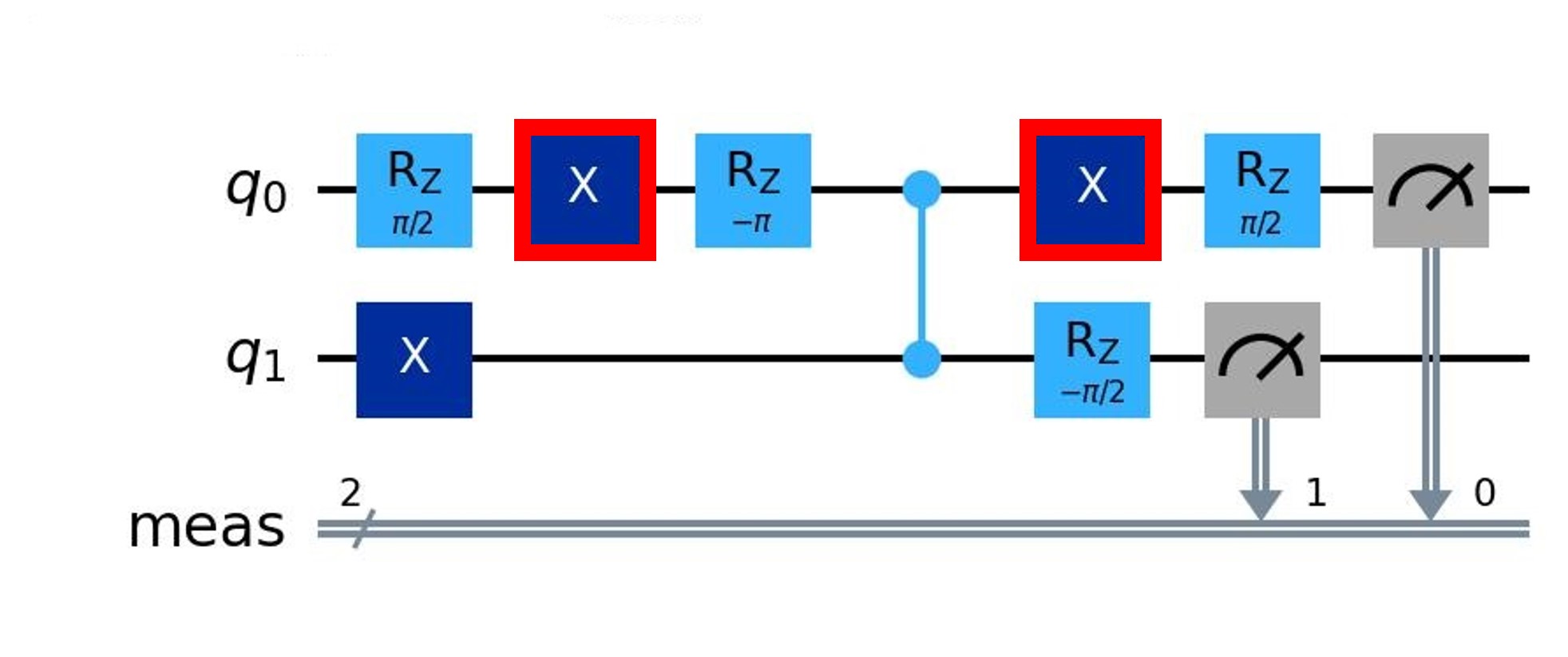}}
    \caption{The NEC is constructed by replacing SX gates in the payload circuit with X gates (outlined in red). The NEC does not generate a superposition by construction and its ideal output when evolving from the ground state is a standard basis state $\ket{k}$ that we can efficiently determine classically.}
    \label{fig:nec_construction_ex}
\end{figure}

Next, we perform randomized compiling to bias the noise towards Pauli. Since the noise estimation circuit has a very similar circuit structure to the payload circuit, we make the assumption that their associated error channels are very close together. \textbf{Assumption 2: the noise vector $\vec{a}$ affecting the NEC is close to the one affecting the payload.} This assumption is supported, by previous investigations on noise estimation circuits \cite{Urbanek_2021MitigDepolNoiseOnQCWithNoiseEstimationCircs} and this paper's quantum hardware results. Limitations and possible improvements of the described NEC for characterization of $A$ are discussed in the next section and in Sec.~\ref{sec:limsAndFutureDir}. 

Since the NEC does not create superposition when evolving from the ground state, it is easy to determine its ideal output classically, which corresponds to the column $\ket{k}$ of $A$ we are characterizing. Let the elements of column $\ket{k}$ of $A$ be denoted by $\vec{b}$. $\vec{b}$ is generated by sampling the NEC. Then, the first column vector $\vec{a}$ can be determined by application of Eq.~\eqref{eq:tilde_a_symmetry}.
% Note that we should pad zero elements so that the basis states for $\vec{a}$ and $\vec{z}$ coincided and their length is a power of 2.
% Let $\tilde \rho=\op{k}$. Then we can characterize column $k$ of $A$ by sampling the noise estimation circuit. 
% Finally, the counts from the payload circuit is $\vec{z}$ and the corrected distribution is calculated using Eq.~\eqref{eq:inversion_fwht}.

\section{Scaling DEM}
Next, we describe scalable methods for solving for $\vec{x}$. We emphasize here that scalable means that DEM  scales efficiently. However, it does not guarantee highly accurate solutions. While transforms are efficient, Eq.~\eqref{eq:inversion_fwht} is only practical for moderate Hilbert space sizes, since the full $2^n$ vectors need to be constructed. For astronomically large Hilbert space sizes (e.g., $2^{100}$), the direct solution using the transforms in Eq.~\eqref{eq:inversion_fwht} is impractical due to time complexity and required storage space of the vectors. There are many techniques from the field of compressed sensing that can be used to address this scenario. We discuss two scalable approaches.

\subsection{Compression}
Assume that the support of $\vec{x}$, $S_x$, is a subset of the support of $\vec{z}$, $S_z$. This assumes that the bitstrings that appear in the ideal distribution also appear in the noisy distribution. Let $\abs{S_z}=K$. Since $S_x\subseteq S_z$, we can compress the problem to 
\begin{align}
    A_S\vec{x}_S=\vec{z}_S,
\end{align}
where $A_S$ is a $K\times K$ submatrix of $A$ over the indices $S_z$, and $\vec{z}_S$ and $\vec{x}_S$ consists of the elements of their corresponding full vectors over the indices $S_z$. This is a much more tractable problem to solve provided that $K\ll 2^n$. If $K$ is small enough, we can fully construct $A_S$ and solve for $\vec{x}_S$. Otherwise, we can use a matrix free iterative solver.

\begin{proposition}\label{prop:compression}
    Let $S_x\subseteq S_z$ and $A$ be invertible. The solution is $\vec{x}=A^{-1}\vec{z}$. Then $\vec{x}_S$ is the only probability vector (unique probability vector solution) such that $A_S\vec{x}_S=\vec{z}_S$. 
\end{proposition}
\begin{proof}
    Assume that there exists a probability vector $\vec{y}_S$ such that $\vec{y}_S\neq \vec{x}_S$ and $A_S\vec{y}_S=A_S\vec{x}_S=\vec{z}_S$. Let $\vec{w}=A\vec{y}$. Note that $A$ is doubly stochastic and hence preserves probability vectors. Thus, $\vec{w}$ must be a probability vector. Since $\vec{w}$ already sums to one over $S_z$, $\vec{w}$ must be zero outside of $S_z$. Thus $\vec{w}=\vec{z}$ and $\vec{y}=\vec{x}$.
\end{proof}
Thus, the compressed problem provides the solution to the full problem provided that $S_x\subseteq S_z$ and $A$ is invertible.

Let $\hat{}$ denote estimation of the variable. In practice, $K_{\hat{z}}=\abs{S_{\hat{z}}}$ is upper bounded by the number of shots $M$ used for sampling, where generally $K_{\hat{z}}\leq M\ll 2^n$.  Also, elements of $\hat{A}_S$ can be efficiently extracted from $\hat{a}$ using the symmetry rule Eq.~\eqref{eq:tilde_a_symmetry}.  In practice, a limitation of compression is shot noise. This can often result in an $\hat{x}$ that is not a probability vector and hence, Prop.~\ref{prop:compression} does not apply. However, in practice, compression generally yields improved distributions.

\subsection{Binning Heuristic}

Note that although $A_S$ is a submatrix of the block circulant matrix $A$, $A_S$ is in general not guaranteed to have a block circulant structure. However, we can assume a XOR convolution structure anyway and use the FWHT. This is the approach used by the binning heuristic. Let $S$ be the support of $\vec{a}$ and $\vec{z}$. Then extend $S$ such that it is the lowest power of two and then apply Eq.~\eqref{eq:inversion_fwht}. This process is essentially a relabeling of the original elements of $S$. Let $f$ be the relabeling function. Then from direct substitution we have that
\begin{align}
    z_{f(i)}=\sum_ja_{f(i\oplus j)}x_{f(j)}.
\end{align}
If $f(i\oplus j)=f(i)\oplus f(j)$, we have
\begin{align}
    z_{f(i)}=\sum_ja_{f(i)\oplus f(j)}x_{f(j)},
\end{align}
which is a XOR convolution over the relabeled indices.

While we cannot generally relabel to a XOR convolution, we can try to bin in such a way that a XOR convolution is possible on a subspace. The heuristic works by trying to construct a contiguous set of bitstrings on the full space. It orders the support in numerical order of bitstrings. Then it iteratively pads bitstrings between the smallest and largest bitstrings. When it is not possible to pad anymore, it inserts elements before and after. It does this until the size of the padded set of bitstrings is a power of two. Note that binning is generally only practical if the power of two is at most around low 20s.

\subsection{Error Bounds}
\label{sec:statstical_analysis}
In this section, we discuss how shot noise affects the corrected distribution and expectation values of observables. For convenience, we can interpret the scenario $\vec{z}=A\vec{x}$ from a measurement or readout error mitigation standpoint, where $A$ is a calibration matrix with the added structure that it is XOR-circulant. Then, the accuracy bounds derived in prior measurement error mitigation papers generally apply.

In measurement error mitigation, Ref.~\cite{Maciejewski_2020MitigOfReadoutNoiseInNearTermQuantDevices} models the measurement apparatus as
\begin{align}
    \textbf{M}^{\text{exp}}=\Lambda\textbf{M}^{\text{ideal}}+\mathbf{\Delta},
\end{align}
where \textbf{M} is block vector of the POVM operators. Thus, $\Lambda$ models the classical errors, $\mathbf{\Delta}$ captures the remaining (non-classical) errors, $\textbf{M}^\text{ideal}$ is the ideal POVM, and $\textbf{M}^\text{exp}$ is the POVM experienced by the payload circuit when executed on quantum hardware. Although $A$ in DEM accounts for more than just measurement errors, it fits mathematically into this POVM model. $A$ is $\Lambda$, $\mathbf{\Delta}$ is due to modeling error, and $\vec{z}=\text{tr}(\tilde\rho \textbf{M}^\text{exp})$ (recall that $\tilde \rho$ is the ideal output of the payload circuit).

In general, there are four main sources of errors: sampling of $\vec{a}$, sampling of $\vec{z}$, modeling, and converting the quasi-probability solution to a valid distribution. In this work, we ignore the modeling error, since tomography of $\vec{a}$ does not inform us of $M^\text{exp}$. First, we will define some important quantities \cite{Maciejewski_2020MitigOfReadoutNoiseInNearTermQuantDevices}. 
Let
\begin{align}
    \epsilon_y=\sqrt{\frac{\ln(2^{n+1}-2)-\ln(P_y)}{2M_y}},
\end{align}
where $n$ is the number of qubits, $M_y$ is the sample size, and  $\text{Pr}(D_{TV}(\hat{y},\vec{y})\leq \epsilon_y)\geq 1-P_y$ for random variable $\vec{y}$. 
Let
\begin{align}
    \norm{A}_{1\rightarrow 1}
\end{align}
denote the maximum $l_1$ norm of the columns of $A$.
% Let $\tilde x=\hat{A}^{-1}\hat{z}$ be the quasi-distribution and $\hat{x}$ be the the final real distribution obtained. Let
% \begin{align}
%     E_A=D_{TV}(\hat{b}, b), 
% \end{align}
% where 
% \begin{align}
%     b= \text{IFWHT}(\vec{1}./\text{FWHT}(\vec{a})),
% \end{align}
% and $\vec{1}$ is the vector of all ones. 
Let
\begin{align}
    E_z=\norm{A^{-1}}_{1\rightarrow 1}\epsilon_{z}
\end{align}
% \begin{align}
%     E_M=\norm{A^{-1}}_{1\rightarrow 1}D_{OP}(\textbf{M}^{\text{exp}}, A\textbf{M}^\text{ideal}),
% \end{align}
and
\begin{align}
    E_C=D_{TV}(\hat{x}, \hat{A}^{-1}\hat{z}).
\end{align}
The first bound gives us the accuracy of DEM corrected distribution given finite sampling of $\vec{a}$ and $\vec{z}$.
\begin{theorem}[Accuracy Bound]\label{thm:accuracyBound}
    \begin{align}
        D_{TV}(\hat{x}, \vec{x})\leq E_C+E_A+E_z,
    \end{align}
    where $E_A$ is due to the sampling error of $\hat{A}$, $E_C$ is the error due to converting the quasiprobability solution to a real distribution, and  $E_z$ is due to the sampling error of $\hat{z}$.
\end{theorem}
\begin{proof}
    From the triangle inequality we have
    \begin{align}
         D_{TV}[\hat{x}, \vec{x}]\leq  D_{TV}(\hat{x}, \hat{A}^{-1}\hat{z})+ D_{TV}(\hat{A}^{-1}\hat{z}, A^{-1}\hat{z})\\
         \notag +D_{TV}(A^{-1}\hat{z}, A^{-1}\vec{z}).
    \end{align}

    Then Ref.~\cite{Maciejewski_2020MitigOfReadoutNoiseInNearTermQuantDevices} shows that $E_C=D_{TV}(\hat{x}, \hat{A}^{-1}\hat{z})$ and $E_z=D_{TV}(A^{-1}\hat{z}, A^{-1}\vec{z})$. Note that the term $E_C$ is directly the $D_{TV}$ since it is directly calculable \cite{Maciejewski_2020MitigOfReadoutNoiseInNearTermQuantDevices}. The only remaining term we have to determine is $D_{TV}(\hat{A}^{-1}\hat{z}, A^{-1}\hat{z})$. From direct calculation, we get
    \begin{align}
        &D_{TV}(\hat{A}^{-1}\hat{z}, A^{-1}\hat z)=\frac{1}{2}\norm{(\hat{A}^{-1}-A^{-1})\hat z}_1\\
        &\leq \frac{1}{2}\norm{\hat{A}^{-1}-A^{-1}}_{1\rightarrow 1}\norm{\hat{z}}_1=\frac{1}{2}\norm{\hat{A}^{-1}-A^{-1}}_{1\rightarrow 1}
        % \frac{1}{2}\norm{A^{-1}(A-\hat{A})\hat{A}^{-1}}_1
    \end{align}
    Then the result follows by setting $E_A=\frac{1}{2}\norm{\hat{A}^{-1}-A^{-1}}_{1\rightarrow 1}.$
    % Then we have that
    % \begin{align}
    %     D_{TV}(\hat{A}^{-1}\hat{z}, A^{-1}\hat z)&\leq \frac{1}{2}\norm{\hat{A}^{-1}-A^{-1}}_1\\
    %     &=\frac{1}{2}\norm{\hat{b}, b}_1=D_{TV}(\hat{b},b),
    % \end{align}
    % where the last line is due to the fact that the inverse of a XOR circulant matrix is also XOR circulant. 
    % Then the result follows by setting $E_A=D_{TV}(\hat{b},b).$ 
\end{proof}
Note that $\frac{1}{2}\norm{\hat{A}^{-1}-A^{-1}}_{1\rightarrow 1}$ can be bounded with further modeling assumptions, but we leave it explicit here. An expansion of the term shows that it scales with $\epsilon_a$.

From Theorem \ref{thm:accuracyBound}, we can write a DEM certificate (which if true we assume that error mitigation is helpful) under an exact XOR assumption. In measurement error mitigation, a useful bound is
\begin{align}
    D_{TV}(\hat{z}, \vec{x})< D_{OP}(M^\text{exp}, M^\text{ideal})
\end{align}
and then the certificate without modeling error is \cite{Maciejewski_2020MitigOfReadoutNoiseInNearTermQuantDevices}
\begin{align}
    E_C+E_A+E_z\leq D_{OP}(M^\text{exp}, M^\text{ideal})+\epsilon_z.
\end{align}
In DEM we use the certificate,
\begin{align}
    E_C+E_A+E_z\leq 1-\hat{a}(0)+\epsilon_a,
\end{align}
where $\hat{a}(0)$ is the first element of the vector.

Lastly, we can bound the error of a class of observables. 
\begin{proposition}[From Ref.~\cite{Bravyi_2021MitigMeasErrsIInMultiQExper}]
    
For bounded observable $O$ with $\abs{O(x)}\leq 1$, the error mitigated value has accuracy $\delta$ with probability $(\geq 2/3)$ for shots $M\geq 4\delta^{-2}\Gamma^2$. $\Gamma^2$ is the DEM overhead and $\Gamma=\norm {A^{-1}}_{1\rightarrow 1}$.
\end{proposition}
\subsection{Quantum Hardware Demonstrations}
The quantum hardware demonstrations were executed on ibm\_marrakesh with 200,000 shots each for the NEC and the payload circuit. We used the binning heuristic for all DEM implementations. The compression method also worked well, but was generally outperformed by binning. 
% The sizes of the probability vectors are truncated to a maximum size of $2^{15}$. 
The corrected distributions are in general quasi-distributions since they can contain negative values. To correct to a near distribution, the method described in Ref.~\cite{Smolin_2012EfficMethForCompTheMaxLikelihoodQuantStFromMeasWithAddGaussNoise} is used. Results for 20-qubit and 30-qubit GHZ state preparation, $\ket{D^{10}_1}$ and $\ket{D^{20}_1}$ state preparation, 6-qubit and 10-qubit quantum phase estimation (QPE), and 5 qubit Grover search are presented. An overview of the fidelities and gate counts of the quantum hardware results are provided in Tables \ref{tab:fidelities} and \ref{tab:gate_counts}, respectively.

\begin{table}[]
    \centering
    \begin{tabular}{|c|c|c|c|c|c|c|c|}
        \hline
        & \multicolumn{7}{c|}{Quantum Hardware Fidelities}\\
        \hline
         & GHZ20 & GHZ30 & $\ket{D^{10}_1}$ & $\ket{D^{20}_1}$ & QPE6 & QPE10 & Grover5\\
         \hline
         DEM & 0.937 & 0.977 & 0.935 &0.803 &0.897 &0.326 &0.749\\
         \hline
         Raw & 0.488 & 0.232 & 0.576 & 0.283 &0.579 &0.029 &0.102\\
         \hline
    \end{tabular}
    \caption{Fidelities for circuits executed on quantum hardware. Raw corresponds to the output of the original circuit without correction. Cor. corresponds to the corrected distribution via DEM. The number after the name specifies the number of qubits used.}
    \label{tab:fidelities}
\end{table}

\begin{table}[]
    \centering
    \begin{tabular}{|c|c|c|c|c|c|c|c|}
        \hline
        &\multicolumn{7}{c|}{Payload Circuit Gate Counts}\\
        \hline
          Op. & GHZ20 & GHZ30 & $\ket{D^{10}_1}$ & $\ket{D^{20}_1}$ & QPE6 & QPE10 & Grover5\\
         \hline
          CZ & 19 & 29& 18& 38& 77& 267& 582\\
         \hline
          RZ & 59 & 89& 45& 95& 91& 259& 742\\
         \hline
          SX & 39 & 59& 50& 100& 158& 556& 1281\\
         \hline
          X & 0 & 0& 1& 6& 5& 9& 48\\
         \hline
         
    \end{tabular}
    \caption{Gate counts for the raw payload circuits executed on quantum hardware. The number after the name specifies the number of qubits used.}
    \label{tab:gate_counts}
\end{table}

\subsubsection{Clifford Circuits}
The following results are for GHZ state preparation circuits which are Clifford. Since the entire circuit is Clifford and Clifford gates map Paulis to Paulis, the composite noise channel affecting the entire circuit is close to Pauli. The corrected distribution for the 30-qubit GHZ quantum hardware execution achieves 97.7\% fidelity and the distributions are shown in Fig.~\ref{fig:ghz_30}.
\begin{figure}
    \centering
    \includegraphics[width=1\linewidth]{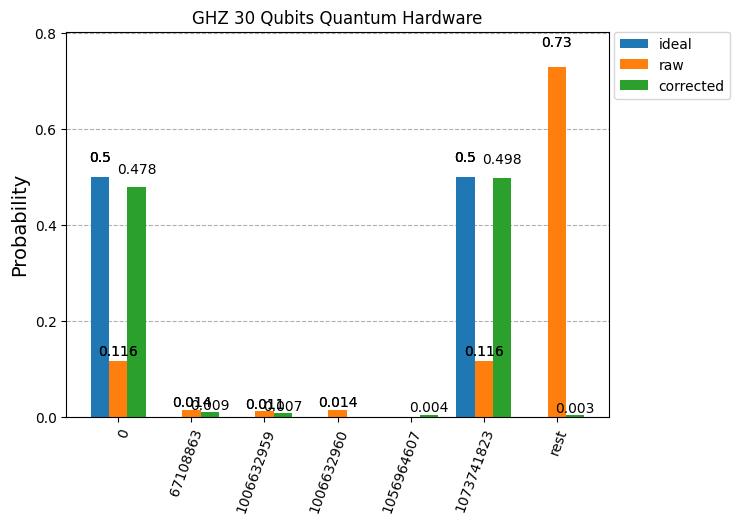}
    \caption{30-qubit GHZ state preparation on quantum hardware. The x-axis uses decimal representation of the basis states.}
    \label{fig:ghz_30}
\end{figure}
The result of the 20-qubit GHZ demonstration is shown in Fig.~\ref{fig:ghz_20}. The corrected distribution fidelity is 93.7\%. Interestingly, the corrected 30-qubit GHZ demonstration achieves a higher fidelity despite the larger number of gates and qubits.

\begin{figure}
    \centering
    \includegraphics[width=\linewidth]{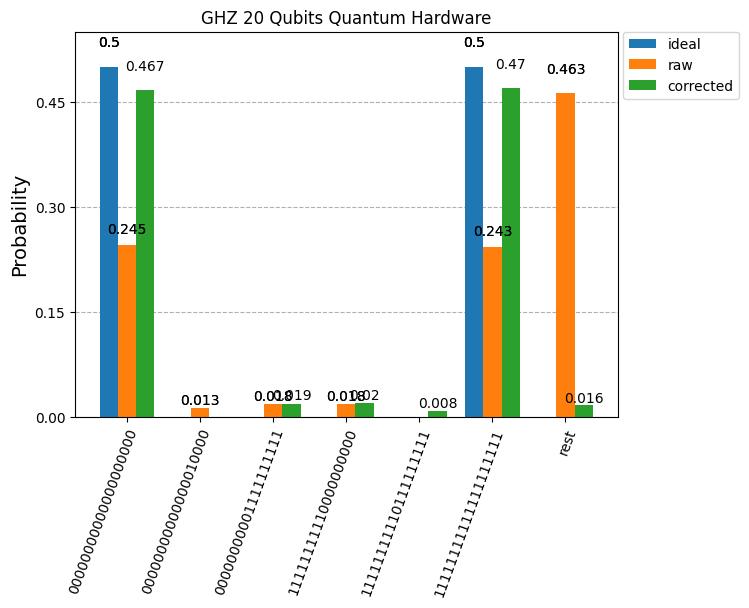}
    \caption{20-qubit GHZ state preparation circuit on quantum hardware.}
    \label{fig:ghz_20}
\end{figure}
% \begin{figure*}
%     \centering
%     \subfloat[]{\includegraphics[width=0.5\linewidth]{figures/plot_ghz20_dist_hardware_ibm_marrakesh.jpg}}
%     \subfloat[]{\includegraphics[width=0.5\linewidth]{figures/plot_ghz20_dist_sim.jpg}}
%     \caption{GHZ 20 qubits executions.}
%     \label{fig:ghz_20}
% \end{figure*}

\subsubsection{Non-Clifford Circuits}
We now turn to non-Clifford circuits. In contrast to Clifford circuits, the presence of non-Clifford gates can destroy the biasing towards Pauli since they can in general map Paulis to non-Paulis. Still, as the following results show, DEM is tolerant to some level of deviation from the Pauli bias. For Dicke states, the sparse state preparation method in Ref.~\cite{gonzales2024efficientsparsestatepreparation} was used to generate the preparation circuits. As shown in Fig.~\ref{fig:dicke_20_1}, the rest bar, which can be interpreted as an approximation of the infidelity, is significantly higher than the corrected rest bar.
\begin{figure}
    \centering
    \includegraphics[width=1\linewidth]{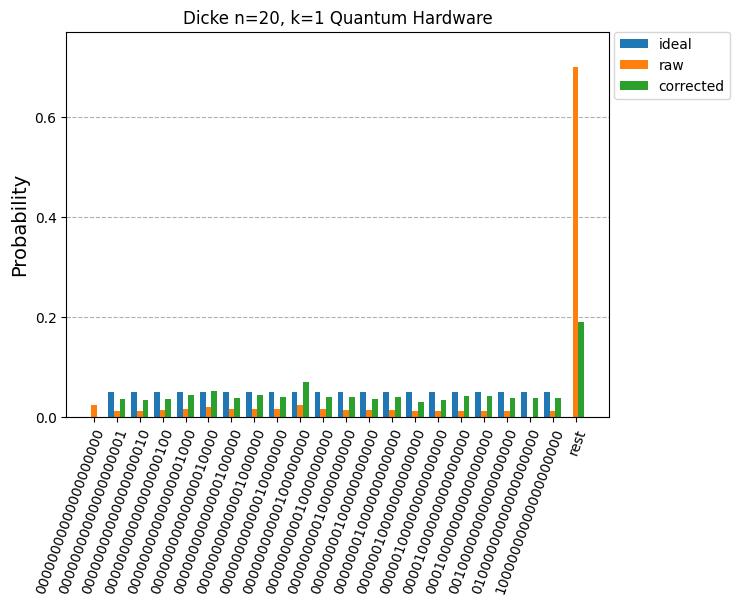}
    \caption{Dicke $\ket{D^{20}_1}$ state preparation on quantum hardware.}
    \label{fig:dicke_20_1}
\end{figure}
Figure \ref{fig:dicke_10_1} shows the results for $\ket{D^{10}_1}$.
\begin{figure}
    \centering
    \includegraphics[width=\linewidth]{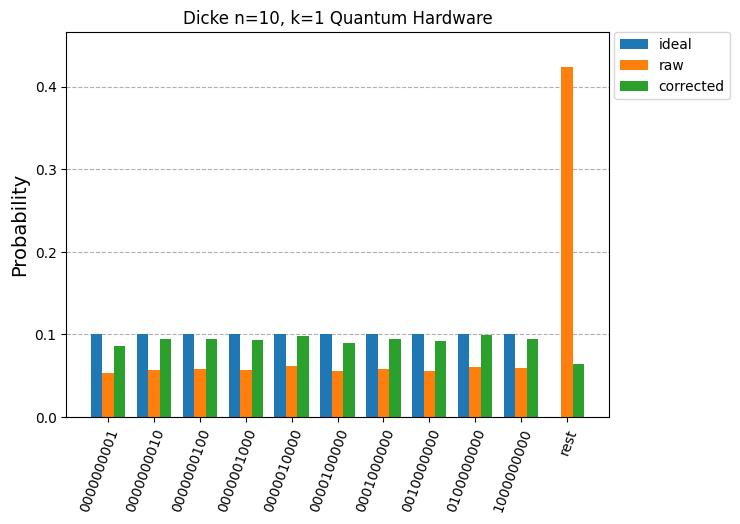}
    \caption{Dicke $\ket{D^{10}_1}$ state preparation on quantum hardware.}
    \label{fig:dicke_10_1}
\end{figure}
% \begin{figure*}
%     \centering
%     \subfloat[]{\includegraphics[width=0.5\linewidth]{figures/plot_dicke10_1_dist_hardware_ibm_marrakesh.jpg}}
%     \subfloat[]{
%     \includegraphics[width=0.5\linewidth]{figures/plot_dicke10_1_dist_sim.jpg}
%     }
%     \caption{$\ket{D^{10}_1}$ state preparation.}
%     \label{fig:dicke_10_1}
% \end{figure*}

Next, the results for the quantum phase estimation demonstrations are shown in Figs.~\ref{fig:qpe6_hardware} and \ref{fig:qpe10_hardware}. The phase angles were chosen to be $\theta=\frac{31}{32}$ and $\theta=\frac{511}{512}$ to correspond with a peak value of all 1s for the 6 and 10 qubit demonstrations, respectively. From Table \ref{tab:fidelities}, QPE10 had the smallest corrected fidelity, but the highest relative increase in fidelity. 
\begin{figure*}
    \centering
    \subfloat[Quantum hardware.\label{fig:qpe6_hardware}]{\includegraphics[width=0.5\linewidth]{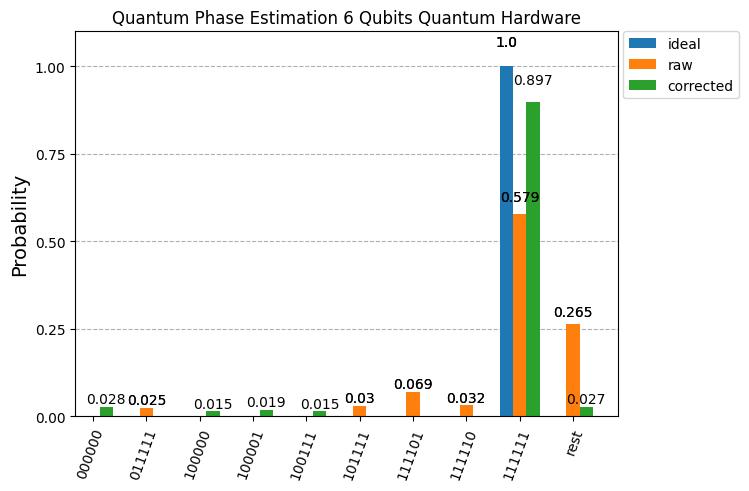}}
    \subfloat[Simulation. \label{fig:qpe6_sim}]{
    \includegraphics[width=0.5\linewidth]{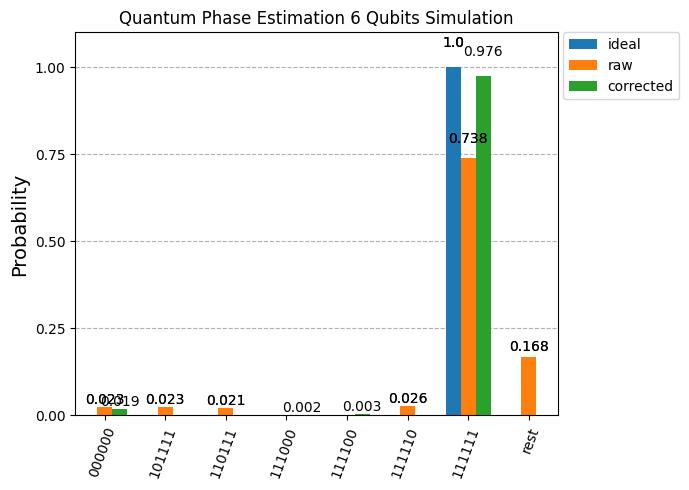}
    }
    \caption{The phase angle used is $\theta=\frac{31}{32}$ which corresponds to an ideal output of all 1s. Note that the control qubit is included in the correction. There is a clear peak at the correct value for both quantum hardware and simulation.}
    \label{fig:QPE_6}
\end{figure*}
\begin{figure*}
    \centering
    \subfloat[Quantum hardware. \label{fig:qpe10_hardware}]{\includegraphics[width=0.5\linewidth]{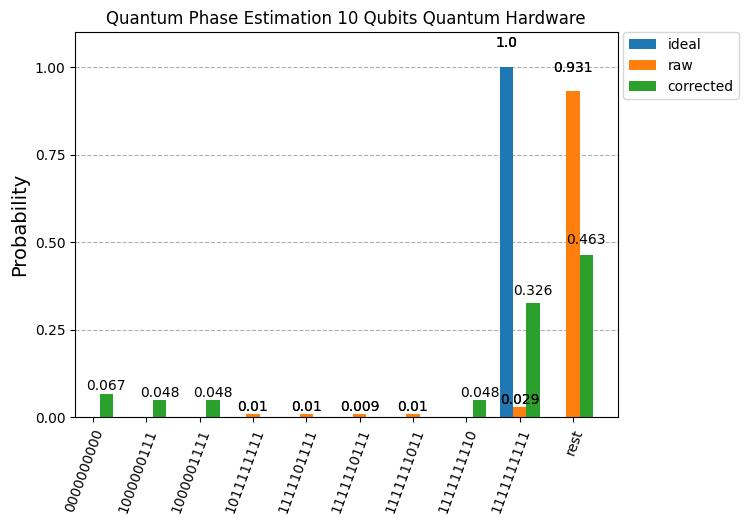}}
    \subfloat[Simulation. \label{fig:qpe10_sim}]{
    \includegraphics[width=0.5\linewidth]{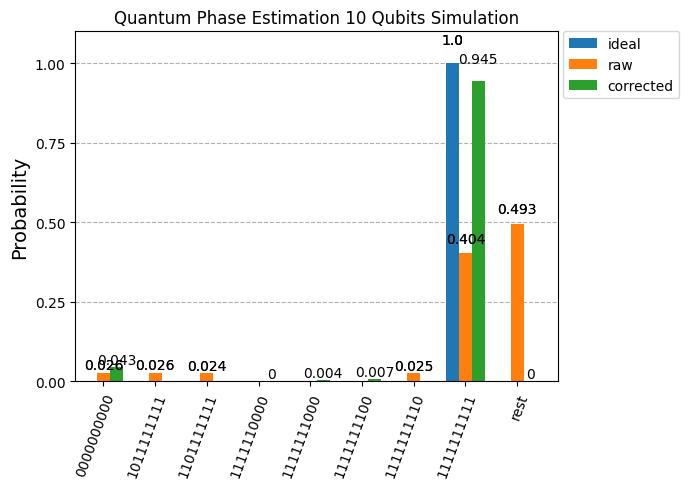}
    }
    \caption{The phase angle used is $\theta=\frac{511}{512}$ which corresponds to an ideal output of all 1s. Note that the control qubit is included in the correction. There is a clear peak at the correct value for both quantum hardware and simulation.}
    \label{fig:QPE_10}
\end{figure*}
Finally, the Grover search demonstration on 5 qubits were conducted with a target state of all 1s. The optimal number of iterations were used. The results are shown Figs.~\ref{fig:grover5_hardware}.
\begin{figure*}
    \centering
    \subfloat[Quantum hardware. \label{fig:grover5_hardware}]{\includegraphics[width=0.5\linewidth]{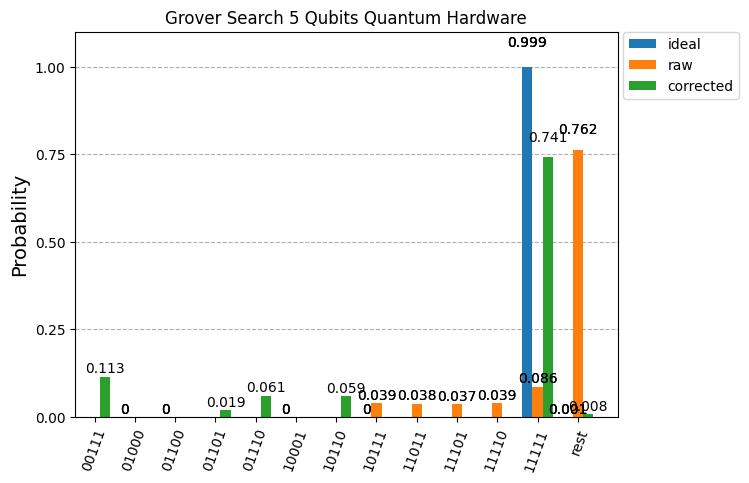}}
    \subfloat[Simulation. \label{fig:grover5_sim}]{
    \includegraphics[width=0.5\linewidth]{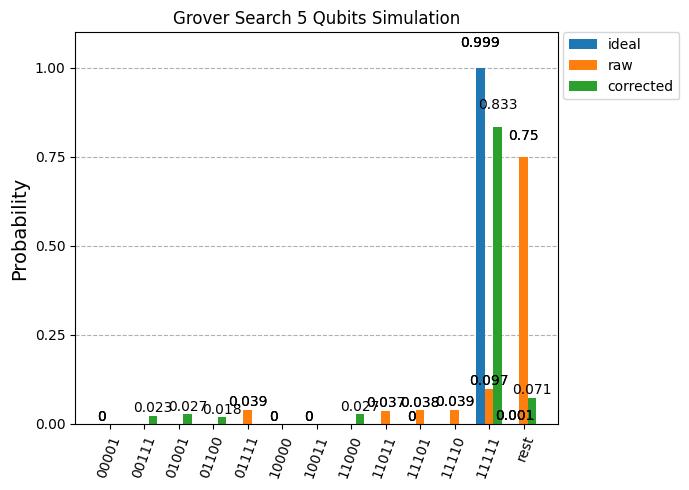}
    }
    \caption{Grover Search. The ideal state is 11111.}
    \label{fig:grover_5}
\end{figure*}

\subsection{Simulations}
Unless stated otherwise, the simulations used 0.001 and 0.01 single qubit and two qubit depolarization rates, respectively. The simulations also assume full connectivity. The QPE results are given in Figs.~\ref{fig:qpe6_sim} and \ref{fig:qpe10_sim}. The Grover search simulation is shown in \ref{fig:grover5_sim}. The same logical circuits were used as their quantum hardware executed counterparts. As expected the classical simulations performed better and in some cases performed much better than the corresponding quantum hardware demonstrations (see Figs.~\ref{fig:qpe10_hardware} and~\ref{fig:qpe10_sim}).

\section{Open Problems and Future Directions}\label{sec:limsAndFutureDir}
There are two key criteria for executing DEM effectively: (1) Pauli bias for the \textit{composite} noise channel (see Eq.~\eqref{eq:pauli-bias}) and (2) accurate characterization of $\vec{a}$. Still, the quantum hardware results show that the approach is resilient to some level of deviation from these criteria. 

For Clifford circuits, (1) is easy to achieve. For non-Clifford circuits, the composite error channel will in general be non-Pauli due to commutation of the Pauli noise channels through non-Clifford gates. A possible solution is magic state injection \cite{Zhou_2000MethodForQuantLogicGateConstruct}. Aside from the the magic state preparation subcircuit, the dynamic circuit is Clifford. However, the behavior of the composite noise channel is unclear, since we are now dealing with an ensemble of circuits.

For criteria (2), with the vanilla NEC method, differences in the noise channels can occur. The improved NEC characterization described in Sec.~\ref{sec:improved_characterization} will probably help in solving this problem. It is also likely that we should use the same twirls for both the payload circuit and NEC.

% Next, the scaling of $\vec{a}$ is handled with truncation. A different approach is to segment $\vec{a}$ into a tensor product. 
Lastly, combining DEM with QECC or other techniques such as noise aware qubit mapping \cite{langfitt2024dynamicresourceallocationquantumErrDetect} are interesting open questions. Note that DEM can probably be applied at the logical level for QECC protected qubits. It is likely that the two approaches together will reduce overheads and enhance their effectiveness.

\section{Conclusions}
This work introduces DEM for correcting the output distribution of a quantum computation. We provide a rigorous theoretical foundation. We prove that a composite Pauli error channel generates an assignment matrix that is described by a XOR convolution between the noise vector and the ideal distribution. Thus, we can correct the noisy output distribution via a FWHT. We introduced a novel noise vector tomography approach that requires only one logical circuit. Additionally, we derived accuracy bounds. The provided DEM implementation biases errors with twirling and requires execution of only 2 logical circuits. The results on quantum hardware show massive fidelity improvements for GHZ and Dicke state preparation, quantum phase estimation, and Grover Search.

The error mitigation techniques of general error mitigation (GEM) \cite{Jattana_2020_generalErrorMitigForQuantCircs, dobler2024scalablegeneralerrormitigation} and measurement error mitigation (MEM) \cite{Maciejewski_2020MitigOfReadoutNoiseInNearTermQuantDevices, Nation_2021ScalableMitigOfMeasErrsOnQuantComp, Bo_2022EffQuantReadErrMitigForSparseMOONQD} use inversion of an assignment matrix. DEM significantly differs from these methods. In MEM, the assignment matrix is only characterized for measurement errors. In GEM, the method explicitly avoids noise biasing. In both of these methods, a circulant structure is not known and the columns of the assignment matrix are independently characterized. Thus, a Transform is not used for inversion. Note that MEM theory corrects exactly for measurement Markovian errors \cite{Geller_2020_rigorousMeasErrCorrec}.

\section{Data Availability}
The data presented in this paper is available online at \url{https://github.com/alvinquantum/quantum_distribution_error_mitigation}.

\section{Acknowledgements}
I thank Daniel Dilley and Zain H. Saleem from Argonne National Laboratory for useful discussions. This material is based upon work supported by Laboratory Directed Research and Development (LDRD) funding from Argonne National Laboratory, provided by the Director, Office of Science, of the U.S. Department of Energy under Contract No. DE-AC02-06CH11357. This research used resources of the Oak Ridge Leadership Computing Facility, which is a DOE Office of Science User Facility supported under Contract DE-AC05-00OR22725.

\appendix

\section{Proofs}
% In the following proofs we make extensive use of the vectorization, $\text{vec}$, operation. For a given basis, $\text{vec}(\sum_{ij}\alpha_{ij}\op{i}{j})=\sum_{ij}\alpha_{ij}\ket{ij}$. We also use the matrix multiplication relation 
% \begin{align}\label{eq:vecMatrixMultRule}
%     \text{vec}(CDE)=(E^T\otimes C)\text{vec}(D)
% \end{align}
% for matrices $C$, $D$, and $E$.

\subsection{Proof of Prop.~\ref{prop:stochastic_relation}}
\begin{proof}
First, we can write the Pauli channel as
\begin{align}
    \sum_i\chi_iP_i\tilde\rho P_i=\sum_{\vec{x}, \vec{z}\in\{0,1\}^n}\chi(\vec{x},\vec{z})X^{\vec{x}}Z^{\vec{z}}\tilde\rho Z^{\vec{z}}X^{\vec{x}},
\end{align}
where if $x_i=1$, $X$ acts on qubit $i$ . If $x_i=0$, $I$ acts on qubit $i$. The notation works likewise for $Z$. Since we are focused on the distribution after measuring in the standard basis, phase errors act like identity. In the standard basis, let the ideal distribution be 
\begin{align}
    a_l=\bra{l}\tilde\rho\ket{l}
\end{align}
and the noisy distribution
\begin{align}
    b_k=\text{tr}(\op{k}\mathcal{E}(\tilde\rho)).
\end{align}
Then,
\begin{align}
    b_k=\sum_{\vec{x}, \vec{z}\in\{0,1\}^n}\chi(\vec{x},\vec{z})\bra{k}X^{\vec{x}}Z^{\vec{z}}\tilde\rho Z^{\vec{z}}X^{\vec{x}}\ket{k}\\
    =\sum_{\vec{x}, \vec{z}\in\{0,1\}^n}\chi(\vec{x},\vec{z})\bra{k\oplus \vec{x}}\tilde\rho \ket{k\oplus \vec{x}}.
\end{align}
Substituting $l=k\oplus \vec{x}$, we get
\begin{align}
    b_k=\sum_{l, \vec{z}\in\{0,1\}^n}\chi(l\oplus k,\vec{z})a_l=\sum_lA_{k,l}a_l,
\end{align}
where $A_{k,l}=\sum_{\vec{z}}\chi(l\oplus k, \vec{z})$.
% $\tilde\rho$ is the ideal quantum state and $\rho'=\sum_i\chi_iP_i\tilde\rho P_i^\dagger$ is the noisy state. Let $\tilde \rho=\sum_{lr}\alpha_{lr}\op{l}{r}$ in the standard basis and measurements be in the standard basis. Measurement in the standard basis yields the output distribution
% \begin{align}
%     P(k|\rho')=\tr(\op{k}\rho')\\=\tr(\op{k}\sum_i\chi_{i}P_i\sum_{lr}\alpha_{lr}\op{l}{r}P_i).
% \end{align}
\end{proof}

\subsection{Proof of Prop.~\ref{prop:AtildeSymmetry}}
\begin{proof}    
Since Pauli channels commute, we have $\mathcal{E}(P_x\rho P_x)=P_x\mathcal{E}(\rho)P_x$, $\forall P_x$, where $P_x$ is a Pauli X string (e.g., $X\otimes I$).  Let $\rho=\op{k}$. Then the probabilities of measuring the noisy state are
\begin{align}
    \notag&\text{tr}(\op{j}\mathcal{E}(P_x\op{k} P_x))\\
    &=\bra{j}\mathcal{E}(P_x\op{k} P_x)\ket{j}=\bra{j}P_x\mathcal{E}(\op{k})P_x\ket{j}, \forall P_x.
\end{align}
Let $\ket{j'}=P_x\ket{j}$ and $\ket{k'}=P_x\ket{k}$. Note that $A_{j,k}=\text{tr}(\op{j}\mathcal{E}(\op{k})).$ Thus, $A_{j,k'}=A_{j',k}$. 
Since this holds for all standard basis states $\ket{j}$, $\ket{k}$, and $P_x$,
\begin{align}
    \bra{j}A P_x\ket{k}=\bra{j}P_xA\ket{k} \quad \forall \ket{j},\ket{k},P_x.
\end{align}
% \begin{align}
%     &\sum_i\chi_iP_iP_x\rho P_xP_i=\sum_i\chi_iP_xP_i\rho P_iP_x,\quad \forall P_x\\
%     &\rightarrow(P_x\otimes P_x)A=A(P_x\otimes P_x),\quad \forall P_x,
% \end{align}
% where $P_x$ is a Pauli X string (e.g., $X\otimes I$). The $a_{iirr}\op{ii}{rr}$ components of $A$ are rearranged by $P_x\otimes P_x$.
% Thus,
% \begin{align}
%     &A P_x
%     = P_xA,\quad \forall P_x\\
%     &\rightarrow\bra{j}A P_x\ket{k}
%     = \bra{j}P_xA\ket{k}\quad\forall \ket{j},\ket{k}, P_x.
% \end{align}
\end{proof}

% \section{Pseudocode}\label{sec:pseudocode}

\subsection{Improved Characterization}\label{sec:improved_characterization} 
In vanilla NEC we replace SX with X. SX maps Y to Z and thus the error channel of the NEC circuit can  significantly deviate from the payload circuit. However, notice that if the SX gates appear at the beginning of the circuit, this scenario does not occur, since there are no gates before them. Ideally, we want to achieve a payload circuit structure of \textbf{SN}, where \textbf{S} consists only of SX gates and \textbf{N} does not contain SX. 
% For real Clifford circuits this is always possible since the real Clifford group is generated by \{Z, H, CZ\} \cite{Hashagen_2018RealRandomizedBenchmarking}.  

In general, we want the gate(s) being replaced for the construction of the NEC circuit to appear in the beginning. For Clifford payload circuits this structure is always possible in the \{H, CX, CZ, P\} basis, since a Clifford operation admits a \textbf{F$_1$HF$_2$} decomposition, where \textbf{F$_i$} cannot create superposition
\cite{Maslov_2018ShorterStabCircsViaBruhatDecompAndQCTransform, Bravyi_2021HadamardFreeCircsExposeTheStructOfTheCliffGroup}. Since we are evolving from the ground state, \textbf{F$_1$} only introduces a global phase and can be replaced with identity.  Since $\textbf{H}$ consists only of H gates, we satisfy the structure.  The NEC is constructed from the payload structure by replacing H, which is the only gate that creates superposition. We leave implementation of the improved characterization for future work.

\vfill

\small

\noindent\framebox{\parbox{0.95\linewidth}{
The submitted manuscript has been created by UChicago Argonne, LLC, Operator of 
Argonne National Laboratory (``Argonne''). Argonne, a U.S.\ Department of 
Energy Office of Science laboratory, is operated under Contract No.\ 
DE-AC02-06CH11357. 
The U.S.\ Government retains for itself, and others acting on its behalf, a 
paid-up nonexclusive, irrevocable worldwide license in said article to 
reproduce, prepare derivative works, distribute copies to the public, and 
perform publicly and display publicly, by or on behalf of the Government.  The 
Department of Energy will provide public access to these results of federally 
sponsored research in accordance with the DOE Public Access Plan. 
http://energy.gov/downloads/doe-public-access-plan.}}\\

%\bibliographystyle{unsrt}
%\bibliography{QEM}

\end{document}